\documentclass[submission,copyright,creativecommons]{eptcs}
\providecommand{\event}{GandALF 2021} 
\usepackage{breakurl}             
\usepackage{underscore}           

\title{The Optimal Way to Play the Most Difficult \\\ Repeated Coordination Games}
\author{Antti Kuusisto
\institute{University of Helsinki and Tampere
University\\ Finland}
\email{antti.kuusisto@helsinki.fi}
\and
Raine R\"{o}nnholm 
\institute{Université Paris-Saclay, CNRS, ENS Paris-Saclay\\
France}
\email{raine.ronnholm@tuni.fi}
}
\def\titlerunning{The optimal way to play the most difficult repeated coordination games}
\def\authorrunning{A. Kuusisto \& R. R\"{o}nnholm}

\usepackage[T1]{fontenc}

\usepackage{amsfonts}
\usepackage{amssymb}
\usepackage{amsmath}
\usepackage{amsthm}
\usepackage{verbatim}
\usepackage{tikz}
\usepackage{graphicx}
\usepackage{hyperref}
\usepackage{xcolor}
\usepackage{xspace}
\usepackage{enumitem}
\usepackage[font=footnotesize]{caption}
\usepackage{vwcol}
\usepackage{arydshln}

\theoremstyle{plain}
\newtheorem{theorem}{Theorem}[section]
\newtheorem{lemma}[theorem]{Lemma}

\newtheorem{proposition}[theorem]{Proposition}
\theoremstyle{definition}
\newtheorem{definition}[theorem]{Definition}
\newtheorem{remark}[theorem]{Remark}
\newtheorem{example}[theorem]{Example}

\newcommand{\aamasp}[1]{{\color{magenta} #1}}
\newcommand{\techrep}[1]{{\color{blue} #1}}
\newcommand{\vcut}[1]{}
\newcommand{\anote}[1]{{\footnotesize \color{blue}\bf{}{AK: #1}}}

\newcommand{\rnote}[1]{{\footnotesize \color{green}\bf{}{RR: #1}}}

\newcommand{\vnote}[1]{}

\newcommand{\Logicname}[1]{\ensuremath{\mathsf{#1}}}

\newcommand{\WLC}{\Logicname{WLC}\xspace}
\newcommand{\CM}{\Logicname{CM}\xspace}
\newcommand{\WM}{\Logicname{WM}\xspace}
\newcommand{\LA}{\Logicname{LA}\xspace}

\newcommand{\ECT}{\Logicname{ECT}\xspace}
\newcommand{\GCT}{\Logicname{GCT}\xspace}
\newcommand{\OSCP}{\Logicname{OSCP}\xspace}

\newcommand{\defstyle}{\textbf}

\begin{document}
\maketitle

\begin{abstract}
This paper investigates repeated win-lose
coordination games (\WLC-games). We analyse which
protocols are optimal for these games,
covering both the worst case and
average case scenarios, i,e., optimizing the
guaranteed and expected coordination times. We begin by
analysing Choice Matching Games (\CM-games) which are a
simple yet fundamental type of \WLC-games, where the
goal of the players is to pick the same choice from a
finite set of initially indistinguishable choices. We
give a fully complete classification of optimal expected and
guaranteed coordination times in two-player \CM-games
and show that the corresponding optimal protocols are
unique in every case---except in the \CM-game
with four choices, which we analyse separately. 

Our results on \CM-games are essential for
proving a more general result on the
difficulty of all \WLC-games: 
we provide a
complete analysis of least upper bounds for
optimal expected coordination times in all
two-player \WLC-games as a function of 
game size. We also show that \CM-games can be
seen as the most difficult games
among all two-player \WLC-games, as they turn out to 
have the greatest optimal expected coordination times.
\end{abstract}

\section{Introduction}
\label{sect:intro}
Pure win-lose coordination games (\WLC-games) are 
simple yet fundamental games  
where all players receive the same payoffs: 1 (win) or 0 (lose).
This paper studies \emph{repeated} \WLC-games, where the players make
simultaneous choices in discrete rounds
until (if ever) succeeding to coordinate on a winning profile.
\emph{Choice matching games} (\CM-games) are the simplest 
class of such games. The choice matching game $\CM_m^n$ has
$n$ players with the goal to choose the same choice 
among $m$ different indistinguishable choices, with no
communication during play. The players 
can use the history of the game (i.e.,
the players' choices in previous rounds) for their benefit as the game proceeds.
For simplicity, we denote the
two-player game $\CM_m^2$\, \nolinebreak by\, \nolinebreak $\CM_m$.
%
%
%
%
%

A paradigmatic real-life scenario with a choice matching game
relates to a phenomenon that has humorously been called
``pavement tango'' or ``droitwich'' in \cite{droit}. Here two people try to 
pass each other but may end up 
blocking each other by repeatedly 
moving sideways into the same direction.
For another example of a choice matching game, consider $\CM_3$, 
the coordination-based 
variant of the rock-paper-scissors game, pictured on the right.
\begin{vwcol}[widths={0.85,0.15}, sep=2mm, justify=flush, rule=0pt, lines=7] 
\indent
Here the
two players (i.e., columns) coordinate if they succeed 
choosing an edge from one of the three rows.
The players first choose randomly; suppose they 
select the nodes in dotted circles. 
%
%
Next the players have three options: (1) to repeat their first choice, (2) to try to coordinate with the first choice of the other player, or (3) to select the choice which has not been selected yet.
Supposing that the players behave symmetrically, (3) is the only way to guarantee coordination in the second round.
%

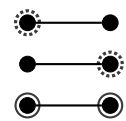
\begin{tikzpicture}[scale=0.55,choice/.style={draw, circle, fill=black!100, inner sep=2.1pt},
	location1/.style={draw, very thick, densely dotted, circle, color=black!77, inner sep=3.2pt},
	location2/.style={draw, thick, circle, color=black!77, inner sep=3.2pt}]
	\node at (0,0) [choice] (00) {};
	\node at (2,0) [choice] (20) {};
	\node at (0,1) [choice] (01) {};	
	\node at (2,1) [choice] (21) {};
	\node at (0,2) [choice] (02) {};	
	\node at (2,2) [choice] (22) {};
	\node at (02) [location1] {};
	\node at (21) [location1] {};
	\node at (00) [location2] {};
	\node at (20) [location2] {};
	\draw[thick] (00) to (20);
	\draw[thick] (01) to (21);
	\draw[thick] (02) to (22);	
%
%
%
	\node at (0,3) {};	
\end{tikzpicture}
%
\end{vwcol}

A general $n$-player \WLC-game is a generalization of $\CM_m^n$ where the
players do not necessarily have to choose from the same row to coordinate,
and it may not even suffice to choose from the same row. In classical 
matrix form representation, two-player choice matching games have ones on
the diagonal and zeroes elsewhere, while general two-player \WLC-games have general
distributions of ones and zeroes;
see Definition~\ref{definition: games} for the full formal details.

In repeated \WLC-games, it is natural to try to coordinate as quickly as possible. 
There are two main scenarios to be investigated: \emph{guaranteeing
coordination} (with certainty) in as few rounds as possible and  
minimizing the \emph{expected number of rounds} for coordination.
The former concerns the number of
rounds it takes to coordinate in the \emph{worst case} and is measured in 
terms of \emph{guaranteed coordination times} ({\GCT}s). The latter
relates to the \emph{average case} analysis measured in
terms of \emph{expected coordination times} ({\ECT}s).


\medskip

\medskip

\noindent
\textbf{Our contributions.} We provide a comprehensive study of upper
bounds for coordination in \emph{all} two-player
repeated \WLC-games, including a classification of related
optimal strategies (called \emph{protocols} in this work).
\WLC-games are represented in a novel way as \emph{relational structures}, which is a key to many of the techniques used in the paper.
\CM-games are central to our work,
being a fundamental class of games and also the 
most difficult games for coordination---in a sense made precise below.

Two protocols play a central role in our study.
We introduce the so-called \emph{loop avoidance} 
protocol \LA (cf. Def. \ref{def: RES}) that essentially instructs players to
play so that the generated history of choices 
always reduces the symmetries (e.g., automorphisms) of the game structure. 
We also use the so-called \emph{wait-or-move}
protocol \WM (cf. Def. \ref{definition: waitormove}),
essentially telling players to randomly
alternate between two choices that both 
coordinate with at least one of the opponent's two choices. We show
that \WM leads to coordination in
\emph{all} two-player \WLC-games \emph{very fast},
the \ECT being at most $3-2p$,
where $p$ is the probability of
coordinating in the first round with random choices.

We then provide a complete analysis of the optimal {\ECT}s and {\GCT}s in all
choice matching games $\CM_m$. We also identify the 
protocols giving the
optimal {\ECT}s and {\GCT}s and show their
\emph{uniqueness}, where possible.
The table in Figure~\ref{fig: summary table} on the next page summarizes these results.
%
%
%
%
%
%
%
%
%
\setlength{\belowcaptionskip}{-15pt}
\begin{figure}
\begin{center}
\scalebox{0.9}{
{\footnotesize
\begin{tabular}{|c||c|c||c|c|}
 \hline
& Optimal \textbf{expected} & Unique optimal & Optimal \textbf{guaranteed} & Unique optimal \\
$m$ & coordination & protocol for & coordination & protocol for \\
& time in $\CM_m$ & expected time & time in $\CM_m$ & guaranteed time \\
\hline \rule{0pt}{1\normalbaselineskip} 
$1$\; & 1 & (any) & 1 & (any) \\[1mm]
$2$ & $2$ & \WM & $\infty$ & --- \\[1mm]
$3$ & $1+\frac{2}{3}$ & \LA & $2$ & \LA \\[1mm]
$4$ & $2+\frac{1}{2}$ & --- & $\infty$ & --- \\[1mm]
$5$ & $2+\frac{1}{3}$ & \LA & $3$ & \LA \\[1mm]
\hdashline \rule{0pt}{1\normalbaselineskip} 
$6$\; & $2+\frac{2}{3}$ & \WM & $\infty$ & --- \\[1mm]
$7$ & $2+\frac{5}{7}$ & \WM & $4$ & \LA \\
\vdots & \vdots & \vdots & \vdots & \vdots \\
$2k$ & $3 - \frac{1}{k}$ & \WM & $\infty$ & --- \\[1mm]
$2k+1$ & $3 - \frac{2}{2k+1}$ & \WM & $k$ & \LA \\[1mm]
\hline
\end{tabular}
}
}
\caption{A complete analysis of two-player choice matching games.}
\label{fig: summary table}
\end{center}
\end{figure}
%
%
%
%
%
%
%
%
%
%
%
%
%
%
%
Our analysis is fully \emph{complete}, as we prove that there 
exists a continuum of optimal protocols for $\CM_4$ and establish that
for all even $m$, no protocol \emph{guarantees} a win in $\CM_m$.

Concerning the more general class of all \WLC-games, we
provide the following complete characterization of
upper bounds for the optimal {\ECT}s in
all two-player $\WLC$-games 
\emph{as a function of game size} 
(a game in a classical matrix form is of size $m$ when the maximum of the
number of rows and columns is $m$):


\medskip

\medskip

\noindent
\textbf{Theorem.}
\emph{For any $m$, the greatest optimal \ECT among 
two-player $\WLC$-games of size $m$ is as follows:}
\begin{center}
\scalebox{0.8}{
\begin{tabular}{|c|c|c|c|}
\hline Game size
& $m\in\mathbb{Z}_+\setminus\{3,5\}$ & \quad$m=5$\quad\text{} & $m=3$ \\
\hline \rule{0pt}{1.2\normalbaselineskip} 
Greatest optimal\; \ECT & $3-\frac{2}{m}$ & $2+\frac{1}{3}$ & $\frac{1+\sqrt{4+\sqrt{17}}}{2}$ $(\approx 1.925)$ \\[1mm]
\hline
\end{tabular}
}
\end{center}

\medskip

Also, concerning two-player \CM-games,
we establish that $\CM_m$ has the
strictly greatest optimal \ECT out of all two-player $\WLC$-games of 
size $m\not = 3$, making \CM-games the most difficult \WLC-games to 
coordinating in. We give a separate full analysis of the case $m=3$.

Two comments are in order. Firstly, in real-life scenarios it can be very inefficient to determine the absolutely optimal protocols taking into account the full game structure. Indeed, this generally requires an analysis of the game structure and thus identifying, e.g., automorphic choices. Computing automorphisms of (even bipartite graphs) is well known to be hard. However, our analysis gives an instant way of finding a fast protocol for any two-player \WLC-game $G$. The \ECT given by this protocol in $G$ is better or equal to the greatest optimal \ECT of the games of the same size (cf. the theorem above).
%
Secondly, our main analysis concerns only two-player games. However, the arguments for this case are already quite involved, so the $n$-player case is expected to be highly complex and is thus left for the future. Furthermore, the two-player case is an especially important special case covering, e.g., learning of communication protocols in distributed systems.

\medskip

\medskip

\noindent
\textbf{Related work.}
Coordination games (see, e.g., \cite{russellcooper},
\cite{Biglaiser1994}) 
are a key topic in game theory, 
with the early foundations laid, inter alia, in the works of
Schelling \cite{schelling} and Lewis \cite{Lewis69}.
Repeated games are---likewise---a key
topic, see for example \cite{mertens}, \cite{AumannMaschler95}, \cite{MailathSamuelson06}.
For seminal work on \emph{repeated coordination games},
see for example the articles \cite{CrawfordHaller95}, \cite{crawfordadaptive},
\cite{lagunoffmatsui}. Repeated coordination games have a
wide range of applications, from 
learning and social choice theory to symmetry breaking in distributed systems.

%
%

However, \WLC-games are a simple class of games
that have not
been extensively studied in the literature. In particular,
choice matching games clearly
constitute a \emph{highly fundamental class of games},
and it is thus surprising
that the analysis of the current paper has not
been previously carried out.
Thus the related analysis is well
justified; it \emph{closes a gap} in the literature.
Indeed, while \CM-games are simple, they precisely capture the highly fundamental coordination problem of \emph{picking the same choice from a set of choices}.

Our study differs from the classical game-theoretic 
work on repeated games where
the focus is on
accumulated payoffs instead of the discrete
outcomes of \WLC-games. Indeed,
repeated \WLC-games are based on 
\emph{reachability objectives}.
Especially our
worst case (but also the average case)  analysis has only superficial overlap with
typical work on repeated games.

However, similar work exists, the most notable example being
the seminal article \cite{CrawfordHaller95} that studies a
generalization of \WLC-games in a framework that has some similarities with our setting. Also the papers \cite{lori2017}, \cite{lorijournal}, \cite{eumas}, \cite{ecai2020} have a
similar focus, as they also 
concentrate on coordination games with discrete win-lose outcomes.
However, the papers \cite{lori2017}, \cite{lorijournal}, \cite{eumas}, \cite{ecai2020} do not investigate optimality of 
protocols in repeated games, unlike the current paper.
The seminal work \cite{CrawfordHaller95} introduces (what is 
equivalent to) the two-player \CM-games in their
final section on general examples.
They also essentially identify the optimal ways of playing $\CM_2$ and $\CM_3$ 
(discussed also in this article) although in a technically
somewhat different setting with
accumulated payoffs. Furthermore, 
they observe that a protocol essentially equivalent to $\WM$ is the best
way to play $\CM_6$, an observation we also make in our
setting. However, optimality of $\WM$ in $\CM_6$ is
not proved in \cite{CrawfordHaller95}. This would require an 
extensive analysis proving that the players cannot make beneficial 
use of asymmetric histories created by non-coordinating choices.
Indeed, the main technical difficulty in our corresponding
setting is to show \emph{uniqueness} of the optimal protocol, which we do in each case where a unique protocol exist.

Nonetheless, despite the differences, the framework of \cite{CrawfordHaller95}
bears some conceptual similarities to ours, e.g., the authors also 
identify \emph{structural protocols} (cf. Definition 
\ref{structuralprotocoldfn} below) as
the natural notion of strategy for studying their framework.
Furthermore, they make extensive use of focal points \cite{schelling} in 
analysing how asymmetric histories can
be used for coordination.
%

Relating to uniqueness of protocols, 
\cite{Goyal96} argues that individual rationality considerations not are not sufficient for ``learning how to coordinate'' in the setting of \cite{CrawfordHaller95}. We agree with \cite{Goyal96} that some
conventions are needed if many protocols lead to the optimal result.
However---in our framework---as we can prove
\emph{uniqueness} of the optimal protocols for $\CM_m$ (for $m\neq 4$), then arguably rational players should adopt precisely these
protocols in \CM-games.

The paper \cite{arxivopt222} is the preprint of the current submission with full proofs and further examples.


\medskip

\medskip

\noindent
\textbf{Techniques used.}
The core of our work
relies on an original approach to
games based on relational structures, as
opposed to using the traditional matrix form representation.
This approach enables us to use graph-theoretic ideas in our arguments.
Both in the worst-case and in average-case analysis, the main 
technical work relies heavily on analysis of symmetries---especially
the way the groups of
automorphisms of games evolve
when playing coordination games.
The most involved result of the worst-case analysis,
Theorem \ref{the: odd m}, is 
proved by reducing the cardinality of the
automorphism group of the \WLC-game studied in a maximally fast fashion.
In the average-case analysis, Theorems \ref{the: 6-ladders},
\ref{the: 5-ladders}, \ref{the: 4-ladders} are proved via a 
combination of analysis of extrema; keeping track of 
groups of automorphisms; graph-theoretic methods; and
focal points \cite{schelling} for breaking symmetry.
The most demanding part here is to show
\emph{uniqueness} of the protocols involved.
Also in the average-case analysis,
Theorem \ref{the: upper bounds for ECTs} relies on earlier theorems and an
extensive and exhaustive analysis of certain bipartite graphs.

In principle one could of course reduce our arguments into the setting with games presented in matrix form. However, then it would become harder to apply
natural graph-theoretic notions like degrees of nodes; special features such as cycles; general symmetries (automorphisms) et cetera. Such notions are key elements in our analysis.
\section{Preliminaries}
\label{sec:prelim}

We define (pure) win-lose coordination games as \emph{relational structures} as follows.

%

\begin{definition}\label{definition: games}
An $n$-player \defstyle{win-lose coordination game (\WLC-game)}
is a relational structure $G=(A,C_1,\dots,C_n,W_G)$ where $A$ is a finite
domain of \defstyle{choices}, each $C_i$ is a non-empty unary relation (representing the choices of player $i$) such that $C_1\cup\cdots\cup C_n=A$,  
and $W_G\subseteq C_1\times\cdots\times C_n$ is an $n$-ary \defstyle{winning relation}.  
For technical convenience, we assume the players have pairwise disjoint choice sets, i.e., $C_i\cap C_j=\emptyset \text{ for all } i \text{ and }j\not=i$.
%
A tuple $\sigma\in C_1\times\cdots\times C_n$ is a \defstyle{choice profile} for $G$
and the choice profiles in $W_G$ are \defstyle{winning (choice) profiles}.
We assume there are no surely losing choices, i.e., choices $c\in A$ that do not belong to any winning choice profile, as rational players would never select such choices. 
%
\end{definition}

We will use the visual representation of \WLC
games as hypergraphs;
two-player games 
become just bipartite graphs under this scheme.
The choices of each player are displayed as
columns of nodes, starting from the choices of
player 1 on the left and ending with the column of choices of player $n$. The
winning relation consists of lines that represent the winning choice profiles. Thus winning choice profiles are also
called \textbf{edges}.
See Example A.1 in \cite{arxivopt222}
for an illustration.



Consider a \WLC-game $G=(A,C_1,\dots,C_n,W_G)$
with $n$ players and $m$ winning 
choice profiles that do not intersect, i.e., none of the $m$ winning choice
profiles share a choice $c\in A$.
Such games form a simple yet 
fundamental class of games, where the goal of
the players is simply to pick the same ``choice'', i.e., to 
simultaneously pick one of the $m$ winning profiles.
These games are called \defstyle{choice matching games}. We let $\CM^n_m$
denote the choice matching game with $n$ players and $m$
choices for each player.
%
%
In this article, we extensively make use of the two-player
choice matching games, $\CM^2_m$.
For these games, we will omit the
superscript ``$2$'' and simply denote them by $\CM_m$.
(Recall the example $\CM_3$
pictured in the introduction.)

Interestingly, out of
all $n$-player \WLC-games where
each of the $n$ players has $m$ choices, the game $\CM^n_m$ has the
least probability of coordination when each player plays randomly. 
In this sense these games can be seen the most difficult for coordination. 
A fully compelling reason for the maximal difficulty
of choice matching games is given later on by Theorem~\ref{hardest games}.

\section{Repeated \WLC-Games}
\label{sect:RWLC}

A \defstyle{repeated play of a \WLC-game $G$} 
consists of consecutive (one-step) plays of $G$.
The repeated play is continued until the
players successfully coordinate, i.e.,
select their choices from a winning choice profile. This
may lead to infinite plays.
We assume that each player can
remember the full history of the
repeated play and use
this information when planning
choices. The history of the play after $k$ rounds is
encoded in a sequence $\mathcal{H}_k$
defined as follows.

\begin{definition}\label{definition: stages}
Let $G$ be an $n$-player \WLC-game.
A pair $(G,\mathcal{H}_k)$ is
called a \defstyle{stage $k$ (or $k$th stage)
in a repeated play of $G$},
where the \defstyle{history} $\mathcal{H}_k$ is a $k$-sequence of
choice profiles in~$G$.
More precisely, $\mathcal{H}_k=\big(H_i\big)_{i\in\{1,\dots,k\}}$
where each $H_i$ is an $n$-ary
relation $H_i=\{(c_1,\dots,c_n)\}$ with a 
single tuple $(c_1,\dots,c_n)\in C_1\times\cdots\times C_n$.
In the case $k=0$, we define $\mathcal{H}_0=\emptyset$. The
stage $(G,\mathcal{H}_0)$ is
the \defstyle{initial stage} (or the $0$th stage).
Like $G$, also $(G,\mathcal{H}_k)$ is a relational structure.
\end{definition}


A stage $k$ contains a history specifying precisely $k$ choice profiles chosen in a repeated play. 
%
%
A winning profile of $(G,\mathcal{H}_k)$ is
called a \defstyle{touched edge} if it
contains some choice $c$ picked in some
round $1,\dots , k$ leading to $(G,\mathcal{H}_k)$.
%
%
%
As we assume that the players only need to coordinate once, we consider repeated plays only up to the first stage where some winning choice profile is selected. If
coordination occurs in the $k$th round, the $k$th stage is called the \defstyle{final stage} of the repeated play.
(But a play can
take infinitely long without coordination.)

\begin{vwcol}[widths={0.85,0.15}, sep=5mm, justify=flush, rule=0pt,lines=5] 
\indent
On the right is a drawing of the stage $2$ in a repeated play of $\CM_2$,
the ``coordination game variant'' of the \emph{matching pennies game} 
(or the ``pavement tango'' from the introduction). 
Here the players have failed to coordinate in round 1 (having picked the choices
with dotted circles) and then failed again by both swapping their choices in
round 2 (solid circles).

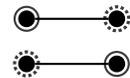
\begin{tikzpicture}[scale=0.6,choice/.style={draw, circle, fill=black!100, inner sep=2.1pt},
	location1/.style={draw, very thick, densely dotted, circle, color=black!77, inner sep=3.2pt},
	location2/.style={draw, thick, circle, color=black!77, inner sep=3.2pt}]
	\node at (0,0) [choice] (00) {};
	\node at (2,0) [choice] (20) {};
	\node at (0,1) [choice] (01) {};	
	\node at (2,1) [choice] (21) {};
	\node at (00) [location1] {};
	\node at (21) [location1] {};
	\node at (01) [location2] {};
	\node at (20) [location2] {};
	\draw[thick] (01) to (21);
	\draw[thick] (00) to (20);	
%
%
    \node at (-0.3,1.7) {};
\end{tikzpicture}
\end{vwcol}

%

%
\scalebox{0.95}{A protocol gives a
\emph{mixed strategy for all stages in all \WLC-games and
for all player roles $i$}:} 

\begin{definition}\label{def: Protocols}
A \defstyle{protocol} $\pi$ is a function outputting a probability distribution $f: C_i\rightarrow[0,1]$ (so $\sum_{c\in C_i}f(c)=1$) with the input of a player $i$ and a stage $(G,\mathcal{H}_k)$ of a repeated \WLC-game.
\end{definition}
Since a protocol can depend on the full history of the current stage, it gives a mixed, memory-based strategy for any repeated \WLC-game. Thus protocols can informally be regarded as global ``behaviour styles" of agents
over the class of all repeated \WLC-games.
It is important to note that all players can see (and remember) the previous choices selected by all the other players---and also the order in which the 
choices have been made.

In the scenario that we study, it is obvious to require that the protocols should act \emph{independently of the names of choices and the names (or ordering) of player roles $i$}.\footnote{Note that if this assumption is not made,  coordination can trivially be guaranteed in a single round in any \WLC-game by using a protocol which chooses some winning choice profile with probability $1$.}
In~\cite{CrawfordHaller95}, this requirement follows from the ``assumption of no common language'' (for describing the game),
and in \cite{lorijournal}, such protocols are
called \emph{structural}.
We shall adopt the
terminology of \cite{lorijournal}.
To extend this concept for repeated games, we
first need to define the notion of a \emph{renaming}. 
The intuitive idea of 
renamings is to extend \emph{isomorphisms} between game 
graphs---including the history---to additionally enable
\emph{permuting the players} $1,\dots , n$ (see Example A.2 in \cite{arxivopt222} for an illustration of the definition).

\begin{definition}\label{def: renamings}
%
A \defstyle{renaming} between
stages $(G,\mathcal{H}_k)$ and $(G',\mathcal{H}_k')$ of $n$-player \WLC-games $G$ and $G'$ is a
pair $(\beta,h)$ where $\beta$ is a permutation of $\{1,\dots , n\}$ and $h$ a
bijection from the domain of $G$ to that of $G'$ such that


\begin{itemize}
\item 
$c\in C_{\beta(i)}\Leftrightarrow h(c) \in C_i'$ for all $i\leq n$ and $c$ in the domain of $G$,

\item 
$(c_1,\dots , c_n) \in W_G
\Leftrightarrow (h(c_{\beta(1)}),\dots , h(c_{\beta(n)}))\in W_{G'}$,

\item 
$(c_1,\dots , c_n) \in H_i
\Leftrightarrow (h(c_{\beta(1)}),\dots , h(c_{\beta(n)}))\in H_i'$ for all $i\leq k$.
\end{itemize}


If $(G,\mathcal{H}_k)$ and $(G',\mathcal{H}_k')$ have the
same domain $A$, we say that $(\beta,h)$ is a \defstyle{renaming of $(G,\mathcal{H}_k)$}.
Choices $c\in C_i$ and $d\in C_j$ are \defstyle{structurally equivalent}, denoted by
$c\sim d$, if there is a renaming $(\beta,h)$ of $(G,\mathcal{H}_k)$ s.t. $\beta(i)=j$ and $h(c)=d$.
%
%
It is easy to see that $\sim$ is an equivalence relation on~$A$. We denote the equivalence class of a choice $c$ by $[c]$.
\end{definition}

\begin{definition}\label{structuralprotocoldfn}
A protocol $\pi$ is \defstyle{structural} if it is indifferent with
respect to renamings,
i.e.,
if $(G,\mathcal{H}_k)$ and $(G',\mathcal{H}_k')$ 
are stages with a renaming $(\beta,\pi)$ between them,
then for any $i$ and any $c\in C_i$, we have 
$
	f(c)=f'(\pi(c)),
$
where $f= \pi((G,\mathcal{H}_k),i)$
and $f' = \pi((G',\mathcal{H}_k'),\beta(i))$.

\end{definition}

%
%
Note that a structural protocol may depend on the full
history which records even the order in which the choices have been played.
%
%
%
%
Hereafter we assume all protocols to be structural.

\begin{definition}\label{def: similarity}
Let $G$ be a \WLC-game and let $S$ and $S'$ be stages of $G$. Let $\sim$
(respectively, $\sim'$) be the structural equivalence relation over $S$ (respectively, $S'$).
We say that $S$ and $S'$ are \defstyle{automorphism-equivalent} if $\sim \,=\, \sim'$.
The stages $S$ and $S'$ are \defstyle{structurally similar} if one can be obtained
from the other by a chain of renamings and automorphism-equvalences.
\end{definition}


%
A choice $c$ in a
stage $S$ is a \defstyle{focal point} if it is not
structurally equivalent to any other choice in
that same stage $S$,
with the possible exception of choices $c'$
that \emph{all belong to some single edge with $c$}.
%
A~focal point breaks symmetry and can be used for winning a 
repeated
coordination
game. This requires that the players have some
(possibly prenegotiated) way to agree on which 
focal point to use. 
See the following example for an illustration.
%
%

\begin{example}\label{ex: focal points}
%
Consider the first two rounds of the game $\CM_5$, pictured below,
where the players fail to coordinate by
first selecting the pair $(a_1,b_2)$ and then fail 
again by selecting the pair $(b_1,c_2)$. 


\begin{center}
\begin{tikzpicture}[scale=0.55,choice/.style={draw, circle, fill=black!100, inner sep=2.1pt},
	location1/.style={draw, very thick, densely dotted, circle, color=black!77, inner sep=3.2pt},
	location2/.style={draw, thick, circle, color=black!77, inner sep=3.2pt}]
	\node at (-3,4) {$\CM_5:$};	
	\node at (0,0) [choice] (00) {};
	\node at (2,0) [choice] (20) {};
	\node at (0,1) [choice] (01) {};
	\node at (2,1) [choice] (21) {};
	\node at (0,2) [choice] (02) {};
	\node at (2,2) [choice] (22) {};
	\node at (0,3) [choice] (03) {};
	\node at (2,3) [choice] (23) {};
	\node at (0,4) [choice] (04) {};
	\node at (2,4) [choice] (24) {};
    \node at (04) [location1] {};
    \node at (23) [location1] {};
    \node at (03) [location2] {};
    \node at (22) [location2] {};
	\draw[thick] (00) to (20);
	\draw[thick] (01) to (21);
	\draw[thick] (02) to (22);
	\draw[thick] (03) to (23);
	\draw[thick] (04) to (24);
	\node at (-0.8,0) {\small $e_1$};
	\node at (-0.8,1) {\small $d_1$};
	\node at (-0.8,2) {\small $c_1$};
	\node at (-0.8,3) {\small $b_1$};	
	\node at (-0.8,4) {\small $a_1$};	
	\node at (2.8,0) {\small $e_2$};
	\node at (2.8,1) {\small $d_2$};
	\node at (2.8,2) {\small $c_2$};
	\node at (2.8,3) {\small $b_2$};	
	\node at (2.8,4) {\small $a_2$};	
\end{tikzpicture}
\end{center}
The structural equivalence classes become 
modified in this scenario as follows:
\begin{itemize}
\item Initially all choices are structurally equivalent.
\item After the first round, the equivalence classes are $\{a_1,b_2\}$, $\{b_1,a_2\}$ and $\{c_1,d_1,e_1,c_2,d_2,e_2\}$.
\item After the second round, the equivalence classes are $\{a_1\}$, $\{a_2\}$, $\{b_1\}$, $\{b_2\}$, $\{c_1\}$, $\{c_2\}$ and $\{d_1,e_1,d_2,e_2\}$.
\end{itemize}

There are no focal points in the initial stage $S_0$ and the
same is true for the next stage $S_1$. However, in the stage $S_2$, all the choices $a_1,b_1,c_1,a_2,b_2,c_2$ become focal points, and the players can thus immediately guarantee coordination in the third round by selecting any winning pair of focal points, i.e., any of the pairs $(a_1,a_2)$, $(b_1,b_2)$, $(c_1,c_2)$.
(We note that, from the point of view of the general study of
rational choice, it may not be obvious which of these
pairs should be selected, so a convention may be needed to fix which protocol to use.)
For another type of example on focal points, see Example A.3 in \cite{arxivopt222}.

\end{example}

In repeated coordination games, it is natural to try to
\emph{coordinate as quickly as possible}.
There are two principal scenarios related to optimizing coordination times:
the \emph{average case} and the \emph{worst case}. The
former concerns the expected
number rounds for coordination and the latter the maximum number in
which coordination can be guaranteed with certainty.

\begin{definition}
Let $(G,\mathcal{H}_k)$ be a stage and let $\pi$ be a protocol.
%
The \defstyle{one-shot coordination probability (\OSCP) from $(G,\mathcal{H}_k)$ with $\pi$} is the probability of
coordinating in a single round from $(G,\mathcal{H}_k)$ when each player follows $\pi$.
%
The \defstyle{expected coordination time (\ECT) from $(G,\mathcal{H}_k)$ with $\pi$} is the
expected value for the number of rounds until
coordination from $(G,\mathcal{H}_k)$ when all players follow $\pi$.
The \defstyle{guaranteed coordination time (\GCT) from $(G,\mathcal{H}_k)$ with $\pi$} is the number $n$ such
that the players are \emph{guaranteed} to coordinate from $(G,\mathcal{H}_k)$ in $n$ rounds, but not in $n-1$ rounds,
when all players follow $\pi$, if such a
number exists. Otherwise this value is $\infty$.
%
%
%

The \OSCP, \ECT and \GCT
from the initial stage $(G,\emptyset)$ with $\pi$ are
referred to as the \OSCP, \ECT and \GCT in $G$ with $\pi$.
We say that $\pi$ is \defstyle{\ECT-optimal} for $G$ if $\pi$
gives the minimum \ECT in $G$, i.e.,  the \ECT given by any protocol $\pi'$ is at
least as large as the one given by $\pi$.
\defstyle{\GCT-optimality} of $\pi$ for $G$ is defined analogously.
\end{definition}



It is possible that there are several different
protocols giving the optimal \ECT (or \GCT)
for a given \WLC-game.
If two protocols $\pi_1$ and $\pi_2$ are both optimal, it may
be that the optimal value is nevertheless \emph{not} obtained when
some of the players follow $\pi_1$ and the
others $\pi_2$.
This leads to a meta-coordination problem about choosing the same
optimal protocol to follow. 
However, such a problem will be avoided if there exists a unique optimal protocol.

\begin{definition}
Let $\pi$ be a protocol and $G$ a \WLC-game.
We say that $\pi$ is \defstyle{uniquely \ECT-optimal} for $G$ if $\pi$ is \ECT-optimal for $G$ and the following holds for all
other protocols $\pi'$ that are \ECT-optimal for $G$: 
for any stage $S$ in $G$ that is reachable with $\pi$, we have $\pi'(S)=\pi(S)$.  
\defstyle{Unique \GCT-optimality} of $\pi$ for $G$ is defined analogously.\footnote{Note that if two different
protocols are uniquely \ECT-optimal for G (and
similarly for unique \GCT-optimality), then their behaviour on $G$
can differ only on stages that are not reachable in the first place by the protcols. Also, their behaviour can of
course differ on games other than $G$.}
\end{definition}

The next lemma states that two structurally similar
stages are essentially the same
stage with respect to different {\ECT}s and {\GCT}s. 
The proof is straightforward.

\begin{lemma}\label{similarity lemma}
Assume stages $S$ and $S'$ of $G$ are structurally similar.
Now, for any protocol $\pi$, 
there exists a protocol $\pi'$ which
gives the same \ECT and \GCT from $S'$ as $\pi$ gives from $S$.
\end{lemma}

\section{Protocols for Repeated \WLC-Games}
\label{sect:solvability}

In this section we introduce two special protocols,
the \emph{loop avoidance protocol} \LA and the
\emph{wait-or-move protocol} \WM. Informally, \LA asserts
that in every round, every player $i$
should avoid---if possible---all choices $c$ that could possibly make
the resulting stage
automorphism-equivalent 
(cf. Def.~\ref{def: similarity}) 
to the current stage, i.e., the stage just before selecting $c$.

\begin{definition}\label{def: RES}
The \defstyle{loop avoidance protocol} (\LA) asserts 
that in every round, every player $i$
should avoid---if possible---all choices $c$ for which the
following condition holds:
if the player $i$ selects $c$, then there exist
choices for the other players so that the 
resulting stage is automorphism-equivalent to
the current stage.
If this condition holds for all choices of the player $i$, then $i$ makes a random choice. Moreover, uniform probability is
used among all the possible choices of $i$.
%
\end{definition}


It is easy to see that \LA avoids, when possible, 
all such stages that are structurally similar to 
\emph{any} earlier stage in the repeated play.
As structurally similar stages are
essentially identical (cf. Lemma \ref{similarity lemma}),
repetition of such stages can be
seen as a ``loop'' in the repeated play. 
When trying to \emph{guarantee} coordination as
quickly as possible, such loops should be avoided.
%
%
%
%
In addition to this heuristic justification,
Theorems \ref{the: even m} and \ref{the: odd m}
give a fully compelling justification for \LA
when considering guaranteed coordination in two-player \CM-games.
For now, we present the following results
(for the proofs, see the correspondingly numbered Propositions 4.2 and 4.3 in \cite{arxivopt222});
see also Example~\ref{ex: use of LA} below
for an illustration of the use of \LA. 
%

\begin{proposition}\label{the: LG3}
$\LA$ is uniquely \ECT-optimal and uniquely \GCT-optimal in $\CM_3$.
\end{proposition}

%

\begin{proposition}\label{the: GCT with RES}
$\LA$ guarantees coordination in games $\CM_m$ in $\lceil m/2 \rceil$ rounds when $m$ is odd, 
but $\LA$ does not guarantee coordination in $\CM_m$ for any even $m$.
\end{proposition}


\begin{example}\label{ex: use of LA}
We illustrate the use of the \LA protocol in the game $\CM_5$, pictured below.
Suppose that coordination fails in the first round. By symmetry, we may assume that the players selected $a_1$ and $b_2$. Now, in the resulting stage $S_1$, the structural equivalence classes are $\{a_1,b_2\}$, $\{b_1,a_2\}$ and $\{c_1,d_1,e_1,c_2,d_2,e_2\}$.

If the pair $(b_1,a_2)$ is selected in the next round, then the structural equivalence classes do not change and thus the resulting next stage is automorphism-equivalent to $S_1$. Hence, by following \LA, player 1 should avoid selecting $b_1$ and player 2 should avoid selecting $a_2$. For the same reason, the players should also avoid selecting the choices $a_1$ and $b_2$.

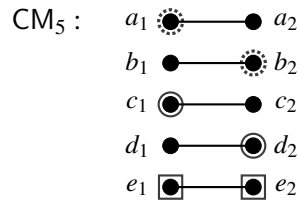
\begin{center}
\begin{tikzpicture}[scale=0.55,choice/.style={draw, circle, fill=black!100, inner sep=2.1pt},
	location1/.style={draw, very thick, densely dotted, circle, color=black!77, inner sep=3.2pt},
	location2/.style={draw, thick, circle, color=black!77, inner sep=3.2pt},
	location3/.style={draw, thick, rectangle, color=black!77, inner sep=4.4pt}]
	\node at (-3,4) {$\CM_5:$};	
	\node at (0,0) [choice] (00) {};
	\node at (2,0) [choice] (20) {};
	\node at (0,1) [choice] (01) {};
	\node at (2,1) [choice] (21) {};
	\node at (0,2) [choice] (02) {};
	\node at (2,2) [choice] (22) {};
	\node at (0,3) [choice] (03) {};
	\node at (2,3) [choice] (23) {};
	\node at (0,4) [choice] (04) {};
	\node at (2,4) [choice] (24) {};
    \node at (04) [location1] {};
    \node at (23) [location1] {};
    \node at (02) [location2] {};
    \node at (21) [location2] {};
    \node at (00) [location3] {};
    \node at (20) [location3] {};
	\draw[thick] (00) to (20);
	\draw[thick] (01) to (21);
	\draw[thick] (02) to (22);
	\draw[thick] (03) to (23);
	\draw[thick] (04) to (24);
	\node at (-0.8,0) {\small $e_1$};
	\node at (-0.8,1) {\small $d_1$};
	\node at (-0.8,2) {\small $c_1$};
	\node at (-0.8,3) {\small $b_1$};	
	\node at (-0.8,4) {\small $a_1$};	
	\node at (2.8,0) {\small $e_2$};
	\node at (2.8,1) {\small $d_2$};
	\node at (2.8,2) {\small $c_2$};
	\node at (2.8,3) {\small $b_2$};	
	\node at (2.8,4) {\small $a_2$};	
\end{tikzpicture}
\end{center}

Hence, by following \LA in $S_1$, the players will select among the set $\{c_1,d_1,e_1,c_2,d_2,e_2\}$ with the uniform probability distribution.
Supposing that they fail again in coordination, we may assume by symmetry that they selected the pair $(c_1,d_2)$.
The equivalence classes in the resulting stage $S_2$ are $\{a_1,b_2\}$, $\{b_1,a_2\}$, $\{c_1,d_2\}$, $\{d_1,c_2\}$ and $\{e_1,e_2\}$. Now, selecting any of the pairs $(a_1,b_2)$, $(b_1,a_2)$, $(c_1,d_2)$ and $(d_1,c_2)$ leads to a next stage which is automorphism-equivalent to~$S_2$. Thus, by following \LA in $S_2$, the players will select the pair $(e_1,e_2)$. This leads to guaranteed coordination in the third round.
\end{example}


We next present the 
\emph{wait-or-move protocol} \WM, which naturally
appears in numerous real-life two-player coordination scenarios. 
Informally, both players
alternate (with equal probability) between two choices:
the players own initial choice and another choice that coordinates 
with the initial choice of the other
player. 
%
%
%

\begin{definition}\label{definition: waitormove}
The \defstyle{wait-or-move protocol (\WM)} for repeated two-player \WLC-games 
goes as follows: first randomly select any choice $c$, and 
thereafter choose, with equal probability, $c$ or a choice $c'$ that
coordinates with the initial choice of the other player (thereby never
picking other choices than $c$ and $c'$).
%
%
\end{definition}

Definition
A.5 in \cite{arxivopt222} specifies
\WM in yet further detail.
The following theorem shows that $\WM$ is very fast in 
relation to {\ECT}s.
This holds for \emph{all} two-player \WLC-games, not only
choice matching games $\CM_m$.
%
%
(The claim is validated by the proof of Proposition 4.5 in \cite{arxivopt222}.)

\begin{theorem}\label{jormatheorem}
Let $G$ be a \WLC-game with
one-shot coordination probability $p$
when both players make their first choice randomly.
Then the expected coordination time by \WM is at most $3-2p$. 
Thus the \ECT with \WM is strictly less than $3$ in
every two-player \WLC-game.
\end{theorem}

%
%
%
%

It follows from the proof of Theorem \ref{jormatheorem} that 
the \ECT with \WM is \emph{exactly} $3-\frac{2}{m}$ in all
choice matching games $\CM_m$. 
Thus the last claim of the theorem \emph{cannot be improved}, as
the {\ECT}s of the games $\CM_m$
grow asymptotically closer to the strict upper bound $3$ 
when $m$ is increased. 
%
%
%
In the particular case of $\CM_2$, the \ECT with \WM is $3-\frac{2}{2}=2$.
Thus the following clearly holds.

\begin{lemma}\label{roskalemma}
When $S =(\CM_m,\mathcal{H}_k)$ is a non-final stage with
exactly two touched edges, then the \ECT from $S$ with \WM is 
2. Moreover, in any \WLC-game $G$, if $S' = (G,\mathcal{H}_k)$ is a non-final stage 
that is reachable by using \WM, then the \ECT from $S'$ with \WM is
at most~2.
\end{lemma}

\vcut{
\begin{proof}
In every round, 
the one-shot coordination probability is $\dfrac{1}{2}$.
Thus the probability of
coordinating in the $m$th round (and not earlier) is
%
%
$(\frac{1}{2})^{m-1} \cdot \frac{1}{2} = (\frac{1}{2})^m$
%
%
for all $m\geq 1$. 
Hence $ECT$ is
%
%
$E\ =\ \sum\limits_{\scriptscriptstyle k\geq 1}^{\scriptscriptstyle\infty} \dfrac{k}{2^{k}}$
%
%
and it is well known that the right hand side sum equals 2.
\end{proof}
}

\WM eventually leads to coordination with asymptotic
probability $1$ in all two-player \WLC-games. But 
it does not
guarantee (with certainty) coordination in any number of rounds
in \WLC-games where the winning relation is not the total relation.
%
%
In a typical real-life scenario, eternal non-coordination is of
course impossible by \WM, but it is conceivable, e.g., that two computing units using the very same pseudorandom number 
generator will never coordinate due to being synchronized to swap their choices in
precisely the same rounds.

It is easy to show that \WM is the unique protocol which
gives the optimal \ECT (namely, 2 rounds) in the ``droitwich-scenario'' 
of the game $\CM_2$ 
(see the proof of Proposition 4.8 in \cite{arxivopt222}).

\begin{proposition}\label{the: LG2}
\WM is uniquely \ECT-optimal in $\CM_2$.
\end{proposition}


Next we compare the pros and cons of \LA and \WM in
two-player \CM-games.
Recall that \WM does not guarantee
coordination in $\CM_m$ (when $m\not=1$),
while \LA does guarantee coordination in $\CM_m$ if and only if $m$ is \emph{odd}.
Concerning \emph{expected} coordination times, it is easy to prove that \WM
gives a smaller \ECT than \LA in $\CM_m$ for all
even $m$ (except for the case $m=2$, where \WM and \LA behave identically).
Thus we now restrict attention to the games $\CM_m$ with odd $m$.
Then, the probability of coordinating in the $\ell$-th
round of $\CM_m$ using \LA, with $\ell\leq \lceil m/2 \rceil$, can
relatively easily be seen to be calculable by 
the formula $P_{\ell,m}$  
defined below (where the product is 
$1$ when $\ell=1$). And using the formula for $P_{\ell,m}$, we also 
get a formula for the expected coordination 
time $E_m$ in $\CM_m$ with $\LA$:
\vspace{-2mm}
\[
P_{\ell,m} = \frac{1}{m - 2\ell + 2}\prod\limits_{k=0}^{k\, =\, {\ell - 2}}\frac{m - 2k - 1}{m - 2k},
\qquad
	E_m =\!\!\! \sum\limits_{\ell=1}^{\ell=\lceil m/2\rceil}\!\!\!\ell\cdot P_{\ell,m}.
\]
%
%
%
%
%
%
Using this and Theorem \ref{jormatheorem}, we 
can compare the {\ECT}s in $\CM_m$ with \LA and \WM for odd~$m$ (see the following table).
%


\begin{center}
\scalebox{0.75
}{
\begin{tabular}{|c|c|c|}
\hline
$m$ & \ECT in $\CM_m$ with \WM & \ECT in $\CM_m$ with \LA \\
\hline
$1$ & 1 & $1$ \\
$3$ & $2+\frac{1}{3}$ & $1+\frac{2}{3}$ \\[1mm]
$5$ & $2+\frac{3}{5}$ & $2+\frac{1}{3}$ \\[1mm]
$7$ & $2+\frac{5}{7}$ & $3$ \\[1mm]
$9$ & $2+\frac{7}{9}$ & $3+\frac{2}{3}$ \\[1mm]
\hline
\end{tabular}
}
\end{center}


\vcut{
We have $E_1 = 1$, 
%
$E_3 = 1\frac{2}{3}$,
%
$E_5 = 2\frac{1}{3}$,
$E_7 = 3$ 
and 
$E_9 = 3\frac{2}{3}$. 
}

Especially the case $m=7$ is
interesting, as the \ECT with \LA is exactly 3 which is precisely the
strict upper bound for the {\ECT}s with \WM 
for the class of all two-player choice matching games $\CM_m$.
Furthermore, $m=7$ is the case 
where \WM becomes
faster than \LA in
relation to {\ECT}s. Thus \WM clearly
stays faster than \LA for all $m\geq 7$, including even values of $m$.

\section{Optimizing Guaranteed Coordination Times}\label{sec: GCTs}

In this section we investigate when coordination can be guaranteed in two-player \CM-games and which protocols give the optimal \GCT for them.
The following result is a direct consequence of the
symmetries of $\CM_m$ for even $m$ (see the
proof of Proposition 5.1 in \cite{arxivopt222}):

\begin{theorem}\label{the: even m}
For all even $m\geq 2$, there is no protocol
guaranteeing coordination in $\CM_m$.
\end{theorem}

We then consider choice matching games $\CM_m$
with an odd $m$. Proposition~\ref{the: GCT with RES} showed that the \GCT with \LA
in these games is $\lceil m/2 \rceil$. 
The next theorem 
shows that this is the optimal \GCT for $\CM_m$, and
moreover, \LA is the unique protocol giving this \GCT. 
The proof of Theorem \ref{the: odd m} is an interesting
exercise in the optimally 
fast elimination of automorphisms.

\begin{theorem}\label{the: odd m}
For any odd $m\geq 1$, $\LA$ is uniquely \GCT-optimal for $\CM_m$.
\end{theorem}

\begin{proof}
Let $m$ be odd. 
Recall that, by Proposition~\ref{the: GCT with RES}, the \GCT in $\CM_m$ with \LA
is $\lceil m/2\rceil$ rounds.
We assume, for contradiction, that there is
some protocol $\pi\neq\LA$ that
guarantees coordination in $\CM_m$ in at
most $\lceil m/2\rceil$ of rounds, possibly less.
As $\pi\neq\LA$, there exists some play of $\CM_m$ where both
players follow $\pi$, and in some round, at
least one of the players chooses a node on a touched edge.
(Recall from the proof of
Proposition~\ref{the: GCT with RES} that \LA never
chooses from a touched edge in \CM-game with an odd number of edges.)
Now, let $S_\ell=(\CM_m,\mathcal{H}_\ell)$
be the first stage of that play 
when this happens---so if $(c,c')$ is the most
recently recorded pair of choices in $S_\ell$, then at 
least one of $c$ and $c'$ is part of an
edge that has already been touched in some earlier round.
And furthermore, in all 
stages $S_{\ell'}$ with $\ell'< \ell$, the most 
recently chosen pair does not contain a choice belonging to an
edge that was touched in some yet earlier round $\ell''<\ell'$.

In the stage $S_{\ell -1}$ it therefore holds that
for every choice profile $(c_i,d_i)$, chosen in some round $i\leq (\ell - 1)$,
the nodes $c_i$ and $d_i$ are structurally equivalent.
Of course also the nodes of $S_{\ell - 1}$ on so far
untouched edges are structurally 
equivalent to each other. Furthermore, the number of
already touched edges in $S_{\ell - 1}$ is the
\emph{even} number $m' = 2(\ell - 1)$.

We will now show that $\pi$ does \emph{not} guarantee a 
win in $\lceil m/2 \rceil - (\ell - 1)$ rounds when starting 
from the stage $S_{\ell - 1}$. This
completes the proof, contradicting the assumption that $\pi$
guarantees a win in $\CM_m$ in at most $\lceil m/2 \rceil$ rounds.

%
Now, recall the stage $S_\ell$ from above where $(c,c')$
contained a choice from an already touched edge. By symmetry, we
may assume that $c$ is such a choice.
Starting from the stage $S_{\ell - 1}$, consider a newly
defined stage $S_\ell'$ where the 
first player again makes the choice $c$ but
the other player this time \emph{makes a structurally equivalent choice $c^*\sim c$}.
This is possible as $\pi$ is a structural protocol. Now note that
the choice profile $(c,c^*)$ is not winning since $c$ and $c^*$ are structurally 
equivalent choices from already touched edges, and thus either $(c,c^*)$
is a choice profile that has already been chosen in some
earlier round $j < \ell$, or the nodes $d,d^*$
adjacent in $\CM_m$ to $c^*,c$ (respectively) form a choice profile $(d,d^*)$
chosen in some earlier round $j < \ell$.
Therefore, in the freshly defined stage $S_\ell'$, the players have in every stage
(including the stage $S_\ell'$ itself) selected a choice profile that consists of two
structurally equivalent choices. Both choices in the most recently selected
choice profile in $S_\ell'$ have been picked from edges that have become 
touched even earlier. It now suffices to
show that it can still take $\lceil m/2 \rceil - (\ell - 1)$ rounds to
finish the game. To see that this is the case, 
we shall next consider a play from the stage $S_\ell'$
onwards where in each remaining round, the choice profile $(e,e^*)$
picked by the players consists of structurally equivalent choices; such a play
exists since $\pi$ is structural.
Due to picking only structurally 
equivalent choices in the remaining play,
when choosing a profile from the already touched 
part, the players will clearly never coordinate. 
And when choosing from the untouched part, immediate coordination is
guaranteed if and only if there is only one untouched edge left.
Therefore the players coordinate exactly when 
they ultimately select from the last untouched edge. As the stage $S_{\ell}'$
has precisely $m - 2(\ell - 1)$ untouched edges, winning in this 
play takes at least
%
%
$\big\lceil\ \frac{m - 2(\ell - 1)}{2}\, \big\rceil\ 
=\ \lceil m/2 \rceil - (\ell - 1)$
rounds to win from $S_\ell'$.
\end{proof}

\section{Optimizing Expected Coordination Times}\label{optimizingsection}

In this section we investigate which protocols give the best {\ECT}s 
for two-player choice matching games. 
We also 
investigate when the best \ECT is obtained by a unique protocol.
We already know by Propositions~\ref{the: LG2} and \ref{the: LG3} that the
optimal {\ECT}s for $\CM_2$ and $\CM_3$
are uniquely given by \WM and \LA, respectively. 
Thus it remains to consider the games $\CM_m$ with $m\geq 4$.
We first cover
the case $m\geq 6$ and show that then \WM is the unique protocol
giving the best \ECT. 
The remaining special cases $m=4$ and $m=5$ will then be examined.
The following auxiliary lemma (proven in \cite{arxivopt222}) will be
used in the proofs.

\smallskip

\begin{lemma}\label{32 Lemma}
The \ECT from $(\CM_m,\mathcal{H}_k)$ with no focal point 
is at least $\frac{3}{2}$ with any protocol.
\end{lemma}

\smallskip

We are now ready to cover the case for choice matching games $\CM_m$ with $m\geq 6$.

\smallskip

\begin{theorem}\label{the: 6-ladders}
\WM is uniquely \ECT-optimal for each $\CM_m$ with $m\geq 6$.
\end{theorem}

\begin{proof}
We first present a formula for estimating the best {\ECT}s in stages of $\CM_m$ (with any $m\geq 1$).
%
%
Let $S := (\CM_m,\mathcal{H}_k)$ be a non-final stage
with exactly two touched edges.
Thus there are $n:=m-2$ untouched edges.
Suppose the players use a protocol $\pi$ behaving as follows in
round $k+1$. Both players pick a choice from
some touched edge with probability $p$ and
from an untouched edge with probability $(1-p)$. A uniform
distribution is used on choices in 
both classes: probability $\frac{p}{2}$ for both
choices on touched edges (which makes
sense by Lemma B.2 of \cite{arxivopt222}) and probability $\frac{1-p}{n}$ for each
choice on untouched edges (which is necessary with a 
structural protocol). If one player selects a choice $c$ from a touched
edge and the other one a choice $c'$ from an untouched
edge, the players win in the next round by choosing the 
edge with~$c'$.
%
%
Note that $c'$ is a focal point, so the winning edge can be
chosen by a structural protocol with probability~1.
(Also other focal points arise which could
alternatively
be used; cf. Example~\ref{ex: focal points}):


Suppose then that $E_1$ is the \ECT with $\pi$ from a stage $(\CM_m,\mathcal{H}_{k+1})$
where both players have chosen a touched
edge in round $k+1$ but failed to coordinate.
Two different such stages $(\CM_m,\mathcal{H}_{k+1})$ exist, but they
are automorphism-equivalent, so $\pi$ can give the same \ECT 
from both of them by Lemma~\ref{similarity lemma}. (Indeed, if $\pi$ gave
two different $\ECT$s, it would make sense to adjust it to give the
smaller one.)
Similarly, suppose $E_2$ is the \ECT with $\pi$ from a
stage $(\CM_m,\mathcal{H}_{k+1}')$ where
both players have chosen an untouched edge in round $k+1$ but failed to coordinate.
Note that all possible such stages $(\CM_m,\mathcal{H}_{k+1}')$
are renamings of each other, so $\pi$
gives the same \ECT from each one.
We next establish that the expected
coordination time from $(\CM_m,\mathcal{H}_k)$ with $\pi$ is now
given by the following formula 
(called formula (E) below):
\begin{align*}
	p^2\Bigl(\frac{1}{2}+\frac{1}{2}\bigl(1+E_1\bigr)\Bigr)\ &+\ 2p(1-p)\cdot 2\ 
	+\ (1-p)^2\, \Bigl(\dfrac{1}{n} + \dfrac{n-1}{n}\Bigl(1 + E_2\Bigr)\Bigr) \tag{E}
\end{align*}

Indeed, both players choose a touched edge in round $k+1$ with probability $p^2$. 
In that case the \ECT from $(\CM_m,\mathcal{H}_{k})$ is 
$\frac{1}{2}+\frac{1}{2}(1+E_1)$,
the first occurrence of $\frac{1}{2}$ corresponding to direct
coordination and the remaining term covering the case where
coordination fails at first.
Both players choose an untouched edge in round $k+1$ with
probability $(1-p)^2$, and then
the \ECT from $(\CM_m,\mathcal{H}_k)$ is
$\frac{1}{n} + \frac{n-1}{n}(1 + E_2).$
\vcut{
Here $\frac{1}{n}$ gives the contribution resulting from direct
coordination and the remaining term covers the case if
coordination fails at first. 
}
The remaining term $2p(1-p)\cdot 2$ is the
contribution of the case
where one player chooses a touched edge and the other player an
untouched one. The probability for this is $2p(1-p)$, and
the remaining factor $2$ indicates that coordination 
immediately happens in the subsequent round $k+2$ using the
focal point created in round $k+1$.

We then present an informal argument sketch for the case $\CM_m$ with $m\geq 6$. We may assume that $E_1\leq 2$ and $E_2\geq\frac{3}{2}$ by Lemmas \ref{roskalemma} and \ref{32 Lemma}.
%
%
Figure \ref{fig: fig} below illustrates the graph of (E) 
with $E_1 = 2$, $E_2 = \frac{3}{2}$, $n=4$, so then (E) has a unique
minimum at $p=1$ when $p\in[0,1]$. This
suggests that---under these parameter values---the players should
always choose a touched edge in
stages with exactly two touched edges. Clearly, lowering $E_1$,
raising $E_2$ or raising $n$ should make it even more beneficial to
choose a touched edge. As we indeed can assume that $E_1\leq 2$
and $E_2\geq\frac{3}{2}$ in $\CM_m$ for $m\geq 6$, this informally 
justifies 
that $\WM$ is uniquely \ECT-optimal in $\CM_m$.
%


%
Next we formalize the argument above.
Let $S := (\CM_m,\mathcal{H}_k)$, $m\geq 6$, be a non-final stage with
precisely two touched edges and $S'$ a stage extending $S$ by one
round where the players both choose an untouched edge but fail to coordinate.
Let $r_1$ (respectively, $r_2$) be the infimum of all
possible {\ECT}s from $S$ (respectively, $S'$) with different protocols.
Note that by Lemma \ref{similarity lemma}, $r_1$ and $r_2$ are independent of
which particular representative stages we
choose, as long as the stages satisfy the given constraints.
Let $\epsilon > 0$ and fix some numbers $E_1$ and $E_2$
such that $|E_1 - r_1 | < \epsilon$ and $|E_2 - r_2 | < \epsilon$. We
assume $E_1\leq 2$
and $E_2\geq\frac{3}{2}$ by Lemmas \ref{roskalemma} and \ref{32 Lemma}.
It is easy to show that with such $E_1$
and $E_2$, the minimum value of the formula (E)
with $p\in [0,1]$ is obtained at $p = 1$ (for any $n = m-2 \geq 4$).

Thus, after the necessarily random choice in round one, the above reasoning shows
the players should choose a
touched edge with probability $p=1$ in
each round. Indeed, assume that the
earliest occasion that a
protocol $\pi_k$ assigns $p\not=1$
occurs in round $k$.
Then, as shown above, the \ECT of $\pi_k$
can be strictly
improved by letting $p=1$ in that round. By Lemma 
B.2 of \cite{arxivopt222}, a uniform probability over the
touched choices should be used. 
\end{proof}

\begin{figure}
\begin{center}
\begin{tikzpicture}[domain=0:1,xscale=5,yscale=10,scale=0.7]
 	\draw[ultra thin,color=gray] (0,1.9) grid[xstep=0.1,ystep=0.05] (1,2.2);
	\draw[->] (0,1.9) -- (1.05,1.9) node[right] {\scriptsize $p$};
	\draw[->] (0,1.9) -- (0,2.23) node[above] {\scriptsize E$(p)$};
	\node at (0,1.87) {\tiny $0$};
	\node at (0.5,1.87) {\tiny $0.5$};
	\node at (1,1.87) {\tiny $1$};
	\node at (-0.07,1.9) {\tiny $1.9$};
	\node at (-0.07,2) {\tiny $2$};
	\node at (-0.07,2.1) {\tiny $2.1$};
	\node at (-0.07,2.2) {\tiny $2.2$};
	\draw[thick] plot (\x, {\x^2*(1/2+(1/2)*(1+2)) + 4*\x*(1-\x) + (1-\x)^2*(1/4+(3/4)*(1+(3/2)))});
\end{tikzpicture}
\;
\begin{tikzpicture}[domain=0:1,xscale=5,yscale=10,scale=0.7]
 	\draw[ultra thin,color=gray] (0,1.6) grid[xstep=0.1,ystep=0.05] (1,1.9);
	\draw[->] (0,1.6) -- (1.05,1.6) node[right] {\scriptsize $p$};
	\draw[->] (0,1.6) -- (0,1.93) node[above] {\scriptsize E$(p)$};
	\node at (0,1.57) {\tiny $0$};
	\node at (0.5,1.57) {\tiny $0.5$};
	\node at (1,1.57) {\tiny $1$};
	\node at (-0.07,1.6) {\tiny $1.6$};
	\node at (-0.07,1.7) {\tiny $1.7$};
	\node at (-0.07,1.8) {\tiny $1.8$};
	\node at (-0.07,1.9) {\tiny $1.9$};
	\draw[thick] plot (\x, {\x^2*(1/2+(1/2)*(1+(3/2))) + 4*\x*(1-\x) + (1-\x)^2*(1/3+(2/3)*(1+1))});
\end{tikzpicture}
\;
\begin{tikzpicture}[domain=0:1,xscale=5,yscale=10,scale=0.7]
 	\draw[ultra thin,color=gray] (0,1.9) grid[xstep=0.1,ystep=0.05] (1,2.1);
	\draw[->] (0,1.9) -- (1.05,1.9) node[right] {\scriptsize $p$};
	\draw[->] (0,1.9) -- (0,2.13) node[above] {\scriptsize E$(p)$};
	\node at (0,1.87) {\tiny $0$};
	\node at (0.5,1.87) {\tiny $0.5$};
	\node at (1,1.87) {\tiny $1$};
	\node at (-0.07,1.9) {\tiny $1.9$};
	\node at (-0.07,2) {\tiny $2$};
	\node at (-0.07,2.1) {\tiny $2.1$};
	\draw[thick] plot (\x, {\x^2*(1/2 + (1/2)*(1+2)) + 4*\x*(1-\x) + (1-\x)^2*(1/2+(1/2)*(1+2))});
\end{tikzpicture}%
\caption{Graph of (E) with\ \ (i)~$n=4$, $E_1=2$, $E_2=\frac{3}{2}$;\ \ \ 
(ii)~$n=3$, $E_1=\frac{3}{2}$, $E_2=1$;\ \ \ 
(iii)~$n=2$, $E_1=E_2=2$.}
\label{fig: fig}
\end{center}
\end{figure}
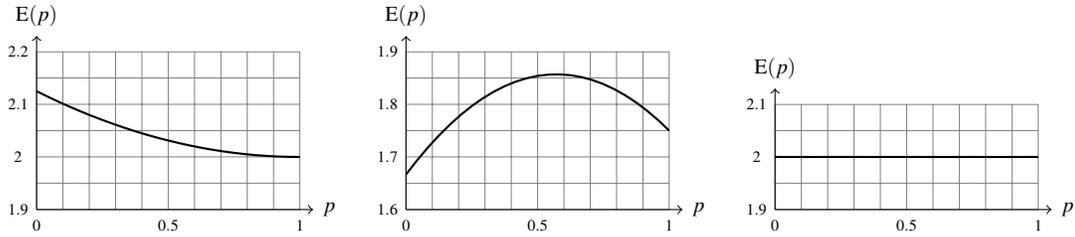%
%
%

We then cover the case for $\CM_5$. The argument is similar to 
the case for $\CM_m$ with $m\geq 6$, 
but this time leads to the use of $\LA$ instead of $\WM$.

\begin{theorem}\label{the: 5-ladders}
For $\CM_5$, \LA is uniquely \ECT-optimal.
\end{theorem}

\begin{proof}
Recall the formula (E) from the proof of Theorem~\ref{the: 6-ladders}.
Let $S := (\CM_5,\mathcal{H}_k)$ be a non-final stage with
precisely two touched edges and $S'$ a stage extending $S$ by one
round where the players both choose an untouched edge but fail to coordinate. The \ECT-optimal protocol from $S'$
chooses the unique winning pair of focal points in round $k+2$, so we
now have $E_2 = 1$.
Let $r_1$ be the infimum of all
possible {\ECT}s from $S$ with different protocols.
Let $\epsilon > 0$ and fix some real number $E_1$
such that $|E_1 - r_1 | < \epsilon$,
assuming $E_1\geq\frac{3}{2}$ (cf. Lemma \ref{32 Lemma}).
It is straightforward to show that with these values,
and with $n=3$, the minimum of (E)
when $p\in [0,1]$ is obtained at $p = 0$. (See
also Figure~\ref{fig: fig} for the graph of (E) when $E_1=\frac{3}{2}$ 
for an illustration. Even then the figure suggests to choose an
untouched edge.)

Thus, after the necessarily random choice in round one, the above reasoning
shows that the players should choose an
untouched edge with probability $1$ in
the second round, thereby following $\LA$.
Coordination is guaranteed (latest) in the third round.
%
\end{proof}


In the last case $m=4$, \WM is \ECT-optimal, but not uniquely, as there
exist infinitely many other \ECT-optimal protocols. 
The reason for this is that---as shown in Figure \ref{fig: fig}---the graph of (E)
becomes the constant line with the value $2$ in the special case where $E_1=E_2=2$, and then any $p\in[0,1]$
gives the optimal value for (E). 
%
%
A complete proof is given in \cite{arxivopt222}.


\begin{theorem}\label{the: 4-ladders}
\WM is \ECT-optimal for $\CM_4$, but there are continuum many other
protocols that are also \ECT-optimal. 
\end{theorem}

We have thus given a \emph{complete analysis} of
optimal {\ECT}s and  {\GCT}s in
two-player \CM-games, summarized in Figure~\ref{fig: summary table}.
Appendix D of \cite{arxivopt222} contains
reflections on the results. We note that the cases for $\CM_m$ with small $m$ are exceptionally important from the point of
view of applications, as such cases tend to occur more 
frequently in real-life scenarios.

\section{The Most Difficult Two-Player \WLC-Games}\label{sec: characterisations}

In this section we give a complete characterization of the upper 
bounds of optimal {\ECT}s in \WLC-games as a
function of game size. 
For any $m\geq 1$, an \defstyle{$m$-choice game}
refers to any 
two-player \WLC-game $G=(A,C_1,C_2,W_G)$ where $m=\max\{|C_1|,|C_2|\}$. 
Note that with the classical matrix representation of an $m$-choice game, the parameter $m$ corresponds to the largest dimension of the matrix.
In this section we will also show 
that $\CM_m$ can be seen as the
\emph{uniquely most difficult} $m$-choice game for all $m\neq 3$, see Theorem~\ref{hardest games}.

Our first theorem shows that the wait-or-move protocol is reasonably ``safe'' to use in 
any $m$-choice game with $m\not\in\{3,5\}$ as it
always guarantees an \ECT which is at most equal
to the upper bound of optimal {\ECT}s of
all $m$-choice games for the particular $m$.

\begin{theorem}\label{safety of WM}
Let $m\not\in\{1,3,5\}$ and consider an
$m$-choice game $G=(A,C_1,C_2,W_G)\not= \CM_m$.
Then the \ECT in $G$ with \WM is strictly smaller than the
optimal \ECT in $\CM_m$.
\end{theorem}

\begin{proof}
By Theorems \ref{the: 6-ladders}, \ref{the: 4-ladders} and Proposition~\ref{the: LG2},
the optimal \ECT in $\CM_m$ is given by \WM. 
We saw in Section \ref{sect:solvability}
that the {\ECT} with \WM is $3-\frac{2}{m}$ in $\CM_m$ 
and at most $3-2p$ in $G$, where $p$ is
the one-shot coordination probability when choosing randomly in $G$.  
Since $G$ is an $m$-choice game, $|W_G|\geq m$. 
If $|W_G|>m$, then $p > \frac{m}{m^2} = \frac{1}{m}$.
And if $|W_G|=m$, we have $p = \frac{m}{mn} = \frac{1}{n} > \frac{1}{m}$
where $n:=\min\{|C_1|,|C_2|\}<m$ since $G\neq\CM_m$.
In both cases, we have $3-2p < 3-\frac{2}{m}$.
\end{proof}


The \defstyle{greatest optimal \ECT} among a class $\mathcal{G}$ of \WLC-games is the value $r$ such that (1) $r$ is the optimal \ECT for some $G\in\mathcal{G}$; and (2) for every $G\in\mathcal{G}$, there is a protocol
which gives it an $\ECT\leq r$. 
By Theorem~\ref{safety of WM}, the greatest optimal \ECT among $m$-choice
games is given by \WM in $\CM_m$ for $m\not\in\{1,3,5\}$.
%
%
The case $m=1$ is trivial, but for the special cases $m=3$ and $m=5$, we need to perform a systematic graph-theoretic analysis of all such $m$-choice games and their {\ECT}s to identify the greatest optimal \ECT among the class. This analysis is done in Appendix~C of~\cite{arxivopt222}, but we sketch below the main ideas starting from the case $m=3$.

Since the optimal \ECT for all $3$-choice games whose graph has a node with degree $3$ is trivially~$1$ round, we may leave them out of the analysis. All the remaining $3$-choice games (up to structural equivalence) are pictured in Figure~\ref{fig: 3-choice games}. Except for $\CM_3$ and the last two games on the right, all of these games have a focal point and thus their optimal \ECT is 1. The optimal \ECT for the second game from the right is $1+\frac{1}{2}$ which is obtained by the protocol making a uniformly random choice every round. The optimal \ECT of the rightmost game is

\begin{center}
$\frac{1}{2}\left(1+\sqrt{4+\sqrt{17}}\right)$
\end{center}

\noindent 
which is indeed the greatest optimal \ECT among all $3$-choice games. (Moreover, this \ECT is obtained by a protocol using highly nontrivial probability distributions for selecting choices.) See Appendix~C.1 of \cite{arxivopt222} for the full details in
each case.

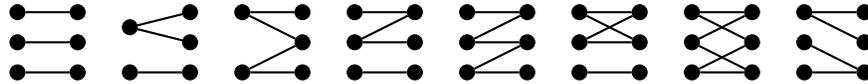
\begin{figure}[htp]
\begin{center}
\begin{tikzpicture}[scale=0.4,choice/.style={draw, circle, fill=black!100, inner sep=2pt}]
	\node at (0,0) [choice] (00) {};
	\node at (2,0) [choice] (20) {};
	\node at (0,1) [choice] (01) {};
	\node at (2,1) [choice] (21) {};
	\node at (0,2) [choice] (02) {};
	\node at (2,2) [choice] (22) {};
	\draw[thick] (00) to (20);
	\draw[thick] (01) to (21);
	\draw[thick] (02) to (22);
\end{tikzpicture}
\quad
\begin{tikzpicture}[scale=0.4,choice/.style={draw, circle, fill=black!100, inner sep=2pt}]
	\node at (0,0) [choice] (00) {};
	\node at (2,0) [choice] (20) {};
	\node at (0,1.5) [choice] (01) {};
	\node at (2,1) [choice] (21) {};
	\node at (2,2) [choice] (22) {};
	\draw[thick] (00) to (20);
	\draw[thick] (01) to (21);
	\draw[thick] (01) to (22);
\end{tikzpicture}
\quad
\begin{tikzpicture}[scale=0.4,choice/.style={draw, circle, fill=black!100, inner sep=2pt}]
	\node at (0,0) [choice] (00) {};
	\node at (2,0) [choice] (20) {};
	\node at (2,1) [choice] (21) {};
	\node at (0,2) [choice] (02) {};
	\node at (2,2) [choice] (22) {};
	\draw[thick] (00) to (20);
	\draw[thick] (00) to (21);
	\draw[thick] (02) to (21);
	\draw[thick] (02) to (22);
\end{tikzpicture}
\quad
\begin{tikzpicture}[scale=0.4,choice/.style={draw, circle, fill=black!100, inner sep=2pt}]
	\node at (0,0) [choice] (00) {};
	\node at (2,0) [choice] (20) {};
	\node at (0,1) [choice] (01) {};
	\node at (2,1) [choice] (21) {};
	\node at (0,2) [choice] (02) {};
	\node at (2,2) [choice] (22) {};
	\draw[thick] (00) to (20);
	\draw[thick] (01) to (21);
	\draw[thick] (01) to (22);
	\draw[thick] (02) to (22);
\end{tikzpicture}
\quad
\begin{tikzpicture}[scale=0.4,choice/.style={draw, circle, fill=black!100, inner sep=2pt}]
	\node at (0,0) [choice] (00) {};
	\node at (2,0) [choice] (20) {};
	\node at (0,1) [choice] (01) {};
	\node at (2,1) [choice] (21) {};
	\node at (0,2) [choice] (02) {};
	\node at (2,2) [choice] (22) {};
	\draw[thick] (00) to (20);
	\draw[thick] (00) to (21);
	\draw[thick] (01) to (21);
	\draw[thick] (01) to (22);
	\draw[thick] (02) to (22);
\end{tikzpicture}
\quad
\begin{tikzpicture}[scale=0.4,choice/.style={draw, circle, fill=black!100, inner sep=2pt}]
	\node at (0,0) [choice] (00) {};
	\node at (2,0) [choice] (20) {};
	\node at (0,1) [choice] (01) {};
	\node at (2,1) [choice] (21) {};
	\node at (0,2) [choice] (02) {};
	\node at (2,2) [choice] (22) {};
	\draw[thick] (00) to (20);
	\draw[thick] (01) to (21);
	\draw[thick] (01) to (22);
	\draw[thick] (02) to (21);
	\draw[thick] (02) to (22);
\end{tikzpicture}
\quad
\begin{tikzpicture}[scale=0.4,choice/.style={draw, circle, fill=black!100, inner sep=2pt}]
	\node at (0,0) [choice] (00) {};
	\node at (2,0) [choice] (20) {};
	\node at (0,1) [choice] (01) {};
	\node at (2,1) [choice] (21) {};
	\node at (0,2) [choice] (02) {};
	\node at (2,2) [choice] (22) {};
	\draw[thick] (00) to (20);
	\draw[thick] (00) to (21);
	\draw[thick] (01) to (20);
	\draw[thick] (01) to (22);
	\draw[thick] (02) to (21);
	\draw[thick] (02) to (22);
\end{tikzpicture}
\quad
\begin{tikzpicture}[scale=0.4,choice/.style={draw, circle, fill=black!100, inner sep=2pt}]
	\node at (0,0) [choice] (00) {};
	\node at (2,0) [choice] (20) {};
	\node at (0,1) [choice] (01) {};
	\node at (2,1) [choice] (21) {};
	\node at (0,2) [choice] (02) {};
	\node at (2,2) [choice] (22) {};
	\draw[thick] (00) to (20);
	\draw[thick] (01) to (20);
	\draw[thick] (02) to (21);
	\draw[thick] (02) to (22);
\end{tikzpicture}
\caption{Game graphs of all (nontrivial) $3$-choice games. The game on the right has the greatest optimal \ECT.}
\label{fig: 3-choice games}
\end{center}
\end{figure}

A similar---yet much more complex---analysis can be done for the case $m=5$ to show that $\CM_5$ has the greatest optimal \ECT ($2+\frac{1}{3}$ rounds) among all $5$-choice games.
(See Appendix~C.2 of \cite{arxivopt222} for the full details of the analysis sketched here.)
We first observe that if $|W_G|>8$, then the \ECT obtained by following \WM in $G$ is $2+\frac{7}{25} < 2+\frac{1}{3}$. 
Next we analyze games $G$ with $|W_G|\leq 8$ based on the degrees of choices. If at least one of the players has a choice of degree at least $4$, we can prove that either there exists a focal point in $G$ or it has the special form where two choices have degree $4$ and all the others have degree~$1$. In this special case we can obtain the \ECT of $2$ rounds by giving the probability $\frac{1}{2}$ for the choices with degree $4$ in the first round and following \WM thereafter.

Suppose next that the degree of all choices in $G$ is at most $3$. If at least one of the players has a choice with degree $3$, we can use graph-theoretic reasoning to show that there is a way to immediately guarantee coordination---unless both players have \emph{exactly one} choice of degree~$3$, denoted by $c$ and $c'$. In this special case, if there is an edge between $c$ and $c'$, they are focal points. Else, $G$ has to be one of the following six games (up to structural equivalence).

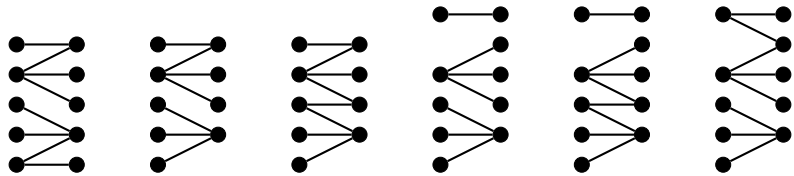
\begin{center}
\begin{tikzpicture}[scale=0.4,choice/.style={draw, circle, fill=black!100, inner sep=2pt}]
%
%
	\node at (0,0) [choice] (00) {};
	\node at (2,0) [choice] (20) {};
	\node at (0,1) [choice] (01) {};
	\node at (2,1) [choice] (21) {};
	\node at (0,2) [choice] (02) {};
	\node at (2,2) [choice] (22) {};
	\node at (0,3) [choice] (03) {};
	\node at (2,3) [choice] (23) {};
	\node at (0,4) [choice] (04) {};
	\node at (2,4) [choice] (24) {};
	\draw[thick] (00) to (21);
	\draw[thick] (01) to (21);
	\draw[thick] (02) to (21);
	\draw[thick] (03) to (22);
	\draw[thick] (03) to (23);
	\draw[thick] (03) to (24);
	\draw[thick] (00) to (20);
	\draw[thick] (04) to (24);
\end{tikzpicture}
\qquad
\begin{tikzpicture}[scale=0.4,choice/.style={draw, circle, fill=black!100, inner sep=2pt}]
	\node at (0,0) [choice] (00) {};
	\node at (0,1) [choice] (01) {};
	\node at (2,1) [choice] (21) {};
	\node at (0,2) [choice] (02) {};
	\node at (2,2) [choice] (22) {};
	\node at (0,3) [choice] (03) {};
	\node at (2,3) [choice] (23) {};
	\node at (0,4) [choice] (04) {};
	\node at (2,4) [choice] (24) {};
	\draw[thick] (00) to (21);
	\draw[thick] (01) to (21);
	\draw[thick] (02) to (21);
	\draw[thick] (03) to (22);
	\draw[thick] (03) to (23);
	\draw[thick] (03) to (24);
	\draw[thick] (04) to (24);
\end{tikzpicture}
\qquad
\begin{tikzpicture}[scale=0.4,choice/.style={draw, circle, fill=black!100, inner sep=2pt}]
	\node at (0,0) [choice] (00) {};
	\node at (0,1) [choice] (01) {};
	\node at (2,1) [choice] (21) {};
	\node at (0,2) [choice] (02) {};
	\node at (2,2) [choice] (22) {};
	\node at (0,3) [choice] (03) {};
	\node at (2,3) [choice] (23) {};
	\node at (0,4) [choice] (04) {};
	\node at (2,4) [choice] (24) {};
	\draw[thick] (00) to (21);
	\draw[thick] (01) to (21);
	\draw[thick] (02) to (21);
	\draw[thick] (03) to (22);
	\draw[thick] (03) to (23);
	\draw[thick] (03) to (24);
	\draw[thick] (02) to (22);
	\draw[thick] (04) to (24);
\end{tikzpicture}
\qquad
\begin{tikzpicture}[scale=0.4,choice/.style={draw, circle, fill=black!100, inner sep=2pt}]
	\node at (0,0) [choice] (00) {};
	\node at (0,1) [choice] (01) {};
	\node at (2,1) [choice] (21) {};
	\node at (0,2) [choice] (02) {};
	\node at (2,2) [choice] (22) {};
	\node at (0,3) [choice] (03) {};
	\node at (2,3) [choice] (23) {};
	\node at (2,4) [choice] (24) {};
    \node at (0,5) [choice] (05) {};
    \node at (2,5) [choice] (25) {};
	\draw[thick] (00) to (21);
	\draw[thick] (01) to (21);
	\draw[thick] (02) to (21);
	\draw[thick] (03) to (22);
	\draw[thick] (03) to (23);
	\draw[thick] (03) to (24);
    \draw[thick] (05) to (25);
\end{tikzpicture}
\qquad
\begin{tikzpicture}[scale=0.4,choice/.style={draw, circle, fill=black!100, inner sep=2pt}]
	\node at (0,0) [choice] (00) {};
	\node at (0,1) [choice] (01) {};
	\node at (2,1) [choice] (21) {};
	\node at (0,2) [choice] (02) {};
	\node at (2,2) [choice] (22) {};
	\node at (0,3) [choice] (03) {};
	\node at (2,3) [choice] (23) {};
	\node at (2,4) [choice] (24) {};
    \node at (0,5) [choice] (05) {};
    \node at (2,5) [choice] (25) {};
	\draw[thick] (00) to (21);
	\draw[thick] (01) to (21);
	\draw[thick] (02) to (21);
	\draw[thick] (03) to (22);
	\draw[thick] (03) to (23);
	\draw[thick] (03) to (24);
	\draw[thick] (02) to (22);
    \draw[thick] (05) to (25);
\end{tikzpicture}
\qquad
\begin{tikzpicture}[scale=0.4,choice/.style={draw, circle, fill=black!100, inner sep=2pt}]
	\node at (0,0) [choice] (00) {};
	\node at (0,1) [choice] (01) {};
	\node at (2,1) [choice] (21) {};
	\node at (0,2) [choice] (02) {};
	\node at (2,2) [choice] (22) {};
	\node at (0,3) [choice] (03) {};
	\node at (2,3) [choice] (23) {};
	\node at (2,4) [choice] (24) {};
    \node at (0,5) [choice] (05) {};
    \node at (2,5) [choice] (25) {};
	\draw[thick] (00) to (21);
	\draw[thick] (01) to (21);
	\draw[thick] (02) to (21);
	\draw[thick] (03) to (22);
	\draw[thick] (03) to (23);
	\draw[thick] (03) to (24);
	\draw[thick] (05) to (24);
    \draw[thick] (05) to (25);
\end{tikzpicture}
\end{center}

Among these games, the leftmost one does not have a focal point unlike all the others, but for that game we can formulate a protocol which gives the \ECT of $2$ rounds.
Finally, supposing that the degrees of all choices in $G$ are at most $2$, the graph of $G$ must consist of \emph{disjoint paths and cycles}. Recalling the assumption $|W_G|\leq 8$, there is a total of 28 such games (listed systematically in Appendix~C.2 of \cite{arxivopt222}). Among those, only the following four do not have a focal point.

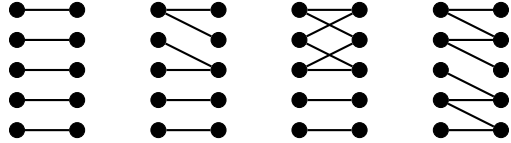
\begin{center}
\begin{tikzpicture}[scale=0.4,choice/.style={draw, circle, fill=black!100, inner sep=2pt}]
	\node at (0,0) [choice] (00) {};
	\node at (2,0) [choice] (20) {};
	\node at (0,1) [choice] (01) {};
	\node at (2,1) [choice] (21) {};
	\node at (0,2) [choice] (02) {};
	\node at (2,2) [choice] (22) {};
	\node at (0,3) [choice] (03) {};
	\node at (2,3) [choice] (23) {};
	\node at (0,4) [choice] (04) {};
	\node at (2,4) [choice] (24) {};
	\draw[thick] (04) to (24);
	\draw[thick] (03) to (23);
	\draw[thick] (02) to (22);
	\draw[thick] (01) to (21);
	\draw[thick] (00) to (20);
\end{tikzpicture}
\qquad
\begin{tikzpicture}[scale=0.4,choice/.style={draw, circle, fill=black!100, inner sep=2pt}]
	\node at (0,0) [choice] (00) {};
	\node at (2,0) [choice] (20) {};
	\node at (0,1) [choice] (01) {};
	\node at (2,1) [choice] (21) {};
	\node at (0,2) [choice] (02) {};
	\node at (2,2) [choice] (22) {};
	\node at (0,3) [choice] (03) {};
	\node at (2,3) [choice] (23) {};
	\node at (0,4) [choice] (04) {};
	\node at (2,4) [choice] (24) {};
	\draw[thick] (04) to (24);
	\draw[thick] (04) to (23);
	\draw[thick] (03) to (22);
	\draw[thick] (02) to (22);
	\draw[thick] (01) to (21);
	\draw[thick] (00) to (20);
\end{tikzpicture}
\qquad
\begin{tikzpicture}[scale=0.4,choice/.style={draw, circle, fill=black!100, inner sep=2pt}]
	\node at (0,0) [choice] (00) {};
	\node at (2,0) [choice] (20) {};
	\node at (0,1) [choice] (01) {};
	\node at (2,1) [choice] (21) {};
	\node at (0,2) [choice] (02) {};
	\node at (2,2) [choice] (22) {};
	\node at (0,3) [choice] (03) {};
	\node at (2,3) [choice] (23) {};
	\node at (0,4) [choice] (04) {};
	\node at (2,4) [choice] (24) {};
	\draw[thick] (04) to (24);
	\draw[thick] (04) to (23);
	\draw[thick] (03) to (24);
	\draw[thick] (03) to (22);
	\draw[thick] (02) to (23);
	\draw[thick] (02) to (22);
	\draw[thick] (01) to (21);
	\draw[thick] (00) to (20);
\end{tikzpicture}
\qquad
\begin{tikzpicture}[scale=0.4,choice/.style={draw, circle, fill=black!100, inner sep=2pt}]
	\node at (0,0) [choice] (00) {};
	\node at (2,0) [choice] (20) {};
	\node at (0,1) [choice] (01) {};
	\node at (2,1) [choice] (21) {};
	\node at (0,2) [choice] (02) {};
	\node at (2,2) [choice] (22) {};
	\node at (0,3) [choice] (03) {};
	\node at (2,3) [choice] (23) {};
	\node at (0,4) [choice] (04) {};
	\node at (2,4) [choice] (24) {};
	\draw[thick] (01) to (20);
	\draw[thick] (01) to (21);
	\draw[thick] (02) to (21);
	\draw[thick] (03) to (22);
	\draw[thick] (03) to (23);
	\draw[thick] (04) to (23);
	\draw[thick] (00) to (20);
	\draw[thick] (04) to (24);
\end{tikzpicture}
\end{center}

However, for the last three games above, we can formulate a protocol which gives them an \ECT of $2$ rounds. Hence $\CM_5$ indeed has the greatest optimal \ECT among all $5$-choice games. 
%

We argue that these graph-theoretic analyses are interesting for their own sake and demonstrate the  usefulness of representing two-player \WLC-games as bipartite graphs.
%
By Theorem~\ref{safety of WM} and the analyses for the 3 and 5-choice cases, we obtain the following results.

\begin{theorem}\label{the: upper bounds for ECTs}
For any $m$, the greatest optimal \ECT among $m$-choice games is given below:

\begin{center}
\scalebox{0.75}{
\begin{tabular}{|c|c|c|c|}
\hline Game size
& $m\in\mathbb{Z}_+\setminus\{3,5\}$ & \quad$m=5$\quad\text{} & $m=3$ \\
\hline \rule{0pt}{1.2\normalbaselineskip} 
Greatest optimal\; \ECT & $3-\frac{2}{m}$ & $2+\frac{1}{3}$ & $\frac{1+\sqrt{4+\sqrt{17}}}{2}$ $(\approx 1.925)$ \\[1mm]
\hline
\end{tabular}
}
\end{center}

\medskip

\end{theorem}

\begin{theorem}\label{hardest games}
For $m\not= 3$, the greatest optimal \ECT among m-choice games is
uniquely realized by $\CM_m$.
\end{theorem}

Hence choice matching games can indeed be seen as \emph{the most difficult} two-player \WLC-games---excluding
the interesting and important
special case of 3-choice games.

\section{Conclusion}


In this paper we gave a complete analysis for two-player \CM-games with respect to both {\GCT}s and {\ECT}s, including 
uniqueness proofs for the related protocols.
We also found optimal upper bounds for optimal {\ECT}s for all two-player \WLC-games when determined according to game size only. 
Moreover, our arguments demonstrate the
usefulness of representing \WLC-games as hypergraphs.
The current paper concentrated on the two-player case as this already turned out a challenging question. A complete characterization of the $n$-player case remains. This is expected to be a highly difficult task that is likely to require sophisticated arguments.



\section*{Acknowledgements}
{We thank Valentin Goranko, Lauri Hella
and Kerkko Luosto for discussions on coordination games.
Antti Kuusisto was supported by the Academy of
Finland project \emph{Theory of computational logics}, grants 324435 and 328987. Raine R\"{o}nnholm was
supported by Jenny and Antti Wihuri Foundation.}




\bibliographystyle{eptcs}
\bibliography{GT-bibliography}


\end{document}